\documentclass[11pt,a4paper]{amsart}

\usepackage[latin1]{inputenc}
\usepackage[english]{babel}
\usepackage{amsmath}

\usepackage{amssymb, mathabx}
\usepackage[numbers]{natbib}
\usepackage{graphicx}
\usepackage{braket}
\usepackage[pdftex,plainpages=false,colorlinks,hyperindex,bookmarksopen,linkcolor=black,citecolor=blue,urlcolor=blue]{hyperref}

\usepackage{mathrsfs}

\usepackage{epstopdf}

\usepackage{braket}

\bibpunct{[}{]}{;}{n}{,}{,}
\usepackage[hmargin=3cm,vmargin={3.5cm,4cm}]{geometry}

\theoremstyle{theorem}
\newtheorem{thm}{Theorem}
\newtheorem{lem}{Lemma}

\theoremstyle{definition}                                 

\theoremstyle{definition}                           

\theoremstyle{remark}                             
\newtheorem*{rmk}{Remark}              

\usepackage{color}

\newcommand{\be}{\begin{eqnarray}}
\newcommand{\ee}{\end{eqnarray}}

\newcommand{\R}{\mathbb{R}}  
\newcommand{\C}{\mathbb{C}} 
\newcommand{\N}{\mathbb{N}} 
\newcommand{\Z}{\mathbb{Z}} 
\def\eg{{\it e.g. }} 
\def\ie{{\it i.e. }}

\newcommand{\blue}[1]{{\color{black}#1}} 

\newcommand{\wt}[1]{\widetilde{#1}}

\newcommand{\ceil}[1]{\lceil #1 \rceil}

\def\eg{{\it e.g.}\ }
\def\ie{{\it i.e.}\ }

\numberwithin{equation}{section}

\allowdisplaybreaks

\begin{document}
\title{On Infinite Order Differential Operators in Fractional Viscoelasticity}
	
	    \author{Andrea Giusti$^\dagger$}
		\address{${}^\dagger$ Department of Physics $\&$ Astronomy, University of 	
    	    Bologna and INFN. Via Irnerio 46, Bologna, ITALY and 
	    	 Arnold Sommerfeld Center, Ludwig-Maximilians-Universit\"at, 
	    	 Theresienstra{\ss}e~37, 80333 M\"unchen, GERMANY.}


\thanks{In \textbf{Fract.~Calc.~Appl.~Anal., Vol~20, No~4 (2017), pp.~854--867}, \textbf{DOI}:~10.1515/fca-2017-0045.}

    \date  {\today}

\begin{abstract}
In this paper we discuss some general properties of viscoelastic models defined in terms of constitutive equations involving infinitely many derivatives (of integer and fractional order). In particular, we consider as a working example the recently developed Bessel models of linear viscoelasticiy that, for short times, behave like fractional Maxwell bodies of order $1/2$.
\end{abstract}

    \maketitle
    

\section{Introduction} \label{Sec-intro}
	In recent years, the quest for potential biomedical applications of ordinary and fractional viscoelastic models has attracted much attention, see \eg \cite{Ata, IC-AG-FM-2017, Freed, AG-FM_MECC16, Mainardi-Gorenflo, Mainardi-12, Meral-2010, Peterson, Pioletti, Provenzano-2002}. In particular, a generalization of the viscoelastic model introduced by Giusti and Mainardi in \cite{AG-FM_MECC16} has lead to the Bessel models \cite{IC-AG-FM-2017}.	
	
	First of all, it is worth recalling some generalities on linear viscoelasticity. As widely discussed in the literature (see \eg \cite{Coleman, Fabrizio-Morro, Gurtin-Sternberg, Mainardi_BOOK10, Mainardi-Spada 2011}), a linear viscoelastic body is defined as a linear system in which the stress plays the role of excitation function for a certain material, while the strain act as the response function, or vice versa. The stress-strain relation, also known as constitutive relation, for a given material can then be expressed mathematically in two different forms, \blue{either as an integral equation or as a differential equation}. 
	
	 First, we denote with $\mathcal{H} ^N$ the Heaviside class of functions (\ie set of causal functions $f(t)$ such that $f \in C^N ([0, + \infty[)$, with $N \in \N$, and with $\sigma, \varepsilon \in \mathcal{H} ^N$ the uniaxial stress and strain functions for a certain system, respectively.	
	
	Given these general assumptions, let us start from the \textbf{integral form}. Then, under the \blue{hypothesis} of sufficiently well behaved causal histories, the constitutive equations can be written in the following forms
	\be \label{eq-integral-form}
	\varepsilon (t) = J(0+) \, \sigma (t) + (\dot{J} \ast \sigma) (t) \, , 
	\qquad
	\sigma (t) = G(0+) \, \varepsilon (t) + (\dot{G} \ast \varepsilon) (t) \, ,
	\ee
	where the $\ast$ represents the convolution product. In this context we call $J (t)$ the \textit{creep compliance} and $G (t)$ the \textit{relaxation modulus}. These function are also often referred to as \textit{material functions of the system}. \blue{It is also worth remarking that both these functions are non-negative on $t \geq 0$, $J(t)$ is non-decreasing and $G(t)$ is non-increasing}.

	Let us now turn our attention to the \textbf{differential form} of the constitutive equation for a given linear system. Indeed, these kind of stress-strain relation are usually encountered while attempting a bottom-up description of a physical system. In fact, it is way easier to write down a differential equation describing the stress-strain relation for a material starting from the fundamental principles of classical mechanics. Moreover, this formalism admits the familiar description of linear viscoelastic systems in terms of \textit{networks of springs and dash-pots}.
	
	In the following we will use $D$ to designate the time derivative operator and we will use $P(D)$ and $Q(D)$ to denote two linear differential operators defined as
	\be 
	P (D) = \sum _{k=0} ^N p_k \, D^k \, , \qquad
	Q (D) = \sum _{k=0} ^N q_k \, D^k \, , 
	\ee
where $p_k , q_k \in \R$, $\forall k \in \N \cup \{ 0 \}$. Clearly, the order of these operators does not exceed $N$.

	Then, the stress-strain relation for a given linear system takes the form
	\be \label{eq-differential-form}
	P (D) \, \sigma (t) = Q (D) \, \varepsilon (t) \, .
	\ee
	
	As shown in \cite{Gurtin-Sternberg}, if the initial conditions \blue{are such that} 
	\be \label{initial-conditions-N}
	\sum _{r=k} ^N p_r \, \sigma ^{(r-k)} (0+) = \sum _{r=k} ^N q_r \, \varepsilon ^{(r-k)} (0+) \, , \qquad k=1, 2 \ldots, N \, ,
	\ee
	together with some regularity conditions on the relaxation function, then the two forms of the constitutive equation are equivalent.
	
	Moreover, if $\sigma (t), \varepsilon (t) = O \left( \exp \left( \alpha \, t \right) \right)$ as $t \to \infty$, for some $\alpha \in \R$, and if \blue{these functions} meet the conditions \eqref{initial-conditions-N}, then, in the Laplace domain, we have that
	\be \label{eq-laplace-form}
	P(s) \, \wt{\sigma} (s) = Q(s) \, \wt{\varepsilon} (s) \, , \qquad \forall s \in \C \,\, : \,\, \texttt{Re} (s) > \alpha \, ,
	\ee
	where we have used the notation $\mathcal{L} \{f(t) \, ; \, s \} := \wt{f} (s)$ with $\mathcal{L}$ denoting the Laplace transform operator. 

	It is worth remarking that $P(s)$ and $Q(s)$ are just polynomials in $s$ whose order does not exceed $N$. Moreover, Eq.~\eqref{eq-laplace-form} states that, if both the stress and the strain meet the condition \eqref{initial-conditions-N}, then the contribution of the initial conditions, while passing from \eqref{eq-differential-form} to \eqref{eq-laplace-form} can be completely neglected.
	
	Now, everything works quite well as long as we work within the framework of a finite linear system, \ie if $N < \infty$. If $N=\infty$ we have to deal with constitutive equation involving infinite order differential operators. Therefore, it is worth recalling some general information about this kind of operators.
	
	Let $\Phi (z) = \sum _{n = 0} ^{\infty} a_{n} \, z^{n}$ and let $f(t)$ be an entire function of $t$, then
	\be \label{IODO}
	\Phi \left( D \right) f(t) := \sum _{n = 0} ^{\infty} a_{n} \, D^n f(t) \, .
	\ee
	In the following $\Phi (D)$ will be referred to as differential operator and the series $ \sum _n a_{n} \, z^{n}$ will be referred to as the \textit{generating power series} of $\Phi (D)$. If the set $\{a_{n} \neq 0 \, | \, n \in \N _{0} \}$ is infinite, then $\Phi (D)$ is said to be a \textit{infinite order} differential operator. In particular, if the series in \eqref{IODO} is convergent $\forall t$, then the differential operator $\Phi (D)$ is said to be \textbf{applicable} to the (entire) function $f(t)$. For further details on these operators, we invite the interested reader to refer to the monograph by Sikkema \cite{Sikkema}. 	
	
	The aim of this paper is to understand what happens if we consider a linear viscoelastic model whose constitutive equation is given by an equation involving infinite order differential operators. Without loss of generality, we will focus our attention on the specific case of a model of the Bessel class \cite{IC-AG-FM-2016}. Nonetheless, the arguments presented in the following sections can be readily extended to any constitutive equation of order $N+\alpha$, with $\alpha \in \R$ and $N$ either finite (integer) or infinite.

\section{The stress-strain equation for a Bessel body} \label{Sec-const-eq}
	Let $\sigma, \varepsilon \in \mathcal{H} ^\infty$, and let us consider a stress-strain relation given by
	\be \label{Bessel-stress-strain-time}
	P_\nu (D_\ast) \, \sigma (t) = Q_\nu (D_\ast) \, \varepsilon (t) \, , \quad \nu > -1 \, ,
	\ee
	where $D_\ast$ denotes the Caputo's fractional derivative and $P_\nu$ and $Q_\nu$ are some infinite order fractional differential operators \blue{defined in terms of the functions}
	\begin{small}
		\be 
	P_\nu (z) = I_\nu (\sqrt{z})  \, , \qquad 
	Q_\nu (z) = I_{\nu+2} (\sqrt{z}) \, ,
	\ee 
	where $I_\alpha (z)$ stands for the series representation of the modified Bessel function of the first kind of order $\alpha \in \R$.
	
	Now, if we assume that the initial data for Eq.~\eqref{Bessel-stress-strain-time} meet the condition
	\be
	\begin{split} \label{initial-condition}
	& \Theta ( \ceil{\nu / 2} - k) \sum _{n = 0} ^\infty p_n \, \sigma ^{(\ceil{\nu / 2} + n - k)} ( 0 +) \, + \\
	& \Theta (k - \ceil{\nu / 2} - 1) \sum _{n = k - \ceil{\nu / 2}} ^\infty p_n \, \sigma ^{(\ceil{\nu / 2} + n - k)} ( 0 +) = \\
	&= 
	\Theta ( \ceil{\nu / 2} + 1 - k) \sum _{n = 0} ^\infty q_n \, \varepsilon ^{(\ceil{\nu / 2} + n + 1 - k)} ( 0 +) \, + \\
	& \Theta (k - \ceil{\nu / 2} - 2) \sum _{n = k - \ceil{\nu / 2} - 1} ^\infty q_n \, \varepsilon ^{(\ceil{\nu / 2} + n + 1 - k)} ( 0 +) \, ,
	\end{split}	
	\ee
		\end{small}	
	for all $k \in \N$, where
	\be 
	p_k = \frac{1}{k! \, \Gamma \left( \nu + k + 1 \right) \, 4^{k + \frac{\nu}{2}}} \, , \qquad 
	q_k = \frac{1}{k! \, \Gamma \left( \nu + k + 3 \right) \, 4^{k + 1 + \frac{\nu}{2}}} \, ,
	\ee
	$\Theta$ is the Heaviside generalized function and $\ceil{x}$ is the celling function of $x$,	then the corresponding constitutive equation, in the Laplace domain, reduces to
	\be \label{Bessel-stress-strain-laplace}
	I_\nu (\sqrt{s}) \, \wt{\sigma} (s) = I_{\nu + 2} (\sqrt{s}) \, \wt{\varepsilon} (s) \, .
	\ee
	Basically, the condition \eqref{initial-condition} guarantees that the contribution these initial data must not appear in the Laplace domain, namely they have to be vanishing or cancel in pair-balance. Moreover, this guarantees the mutual consistency of both the integral representation of the stress-strain relation and the differential form.
	
	Now, from Eq.~\eqref{Bessel-stress-strain-laplace} one can easily deduce the creep compliance for the model (in the Laplace domain), that turns out to be exactly the one that defines the Bessel models \cite{IC-AG-FM-2016}. Indeed, we have that
	\be
	s \wt{J} (s; \nu) = \frac{I_\nu (\sqrt{s})}{I_{\nu + 2} (\sqrt{s})} \, .
	\ee 	
The Bessel creep compliance, in the time domain, is a Bernstein function given by (see \cite{IC-AG-FM-2016, AG-FM-EPJP})
\begin{small}
\begin{equation}
\begin{split}
J(t; \, \nu) = &2 \left( \frac{\nu + 2}{\nu + 3} \right) + 4 (\nu + 1) (\nu + 2) t\\
&- 4 (\nu + 1) \sum _{n=1} ^{\infty} \frac{1}{j_{\nu + 2 , \, n} ^{2}} \exp \left( - j_{\nu + 2 , \, n} ^{2} \, t \right) \, , 
\end{split}
\end{equation} 
\end{small}
where $j_{\nu + 2 , \, n}$ represents the $n$th positive real root of the Bessel functions  of the corresponding order. Moreover, it is worth remarking (see \cite{AG-FM-EPJP}) this function can be thought of as the generating function for an infinite network of \textit{ordinary} springs and dash-pots.  

	Furthermore, this last remark is also supported by the following argument. If the stress and the strain histories meet the condition \eqref{initial-condition}, then we have that \eqref{Bessel-stress-strain-time} reduces to an \textit{ordinary} infinite order differential equation. 
Indeed, we can always recast \eqref{Bessel-stress-strain-laplace} as 
$$
\overline{P} _\nu (s) \, \wt{\sigma} (s) 
	= 
\overline{Q} _\nu (s) \, 
\wt{\varepsilon} (s) \, ,	
$$
where
	\blue{
	\be 
	\overline{P} _\nu (s) = 
	\left( \frac{4}{s} \right) ^{\frac{\nu}{2}} \, I _{\nu} (\sqrt{s})	
	=
	\sum _{k = 0} ^\infty \frac{1}{k! \, \Gamma \left( \nu + k + 1 \right)} \, \left( \frac{s}{4} \right) ^{k}  \, , \label{op_P_nu_reg}
	\ee
	\be
	\overline{Q} _\nu (s) = 
	\left( \frac{4}{s} \right) ^{\frac{\nu}{2}} \, I _{\nu + 2} (\sqrt{s})	
	= 
	\sum _{k = 0} ^\infty \frac{1}{k! \, \Gamma \left( \nu + k + 3 \right)} \, \left( \frac{s}{4} \right) ^{k + 1} \, . \label{op_Q_nu_reg}
	\ee
	}

	Then, inverting back to the time domain, we get the new constitutive equation
	\be \label{Bessel-stress-strain-time-regularized}
	\overline{P} _\nu (D) \, \sigma (t) = \overline{Q} _\nu (D) \, \varepsilon (t) \, , \quad \nu > -1 \, ,
	\ee
	where $D$ is, again, the \textit{ordinary} derivative with respect to time.
	
	\vskip 0.5 cm	
	
	To prove all these results we just need to focus our attention on the behaviour of the initial conditions. In the following discussion we will generalize the arguments presented in \cite{Gurtin-Sternberg} concerning the treatment of initial conditions in linear viscoelasticity to the case of infinite order differential operators.  

	First, let us assume $P_\nu (D_\ast) \, \sigma (t)$ and $Q_\nu (D_\ast) \, \varepsilon (t)$ to be convergent $\forall t \geq 0$. If this is the case, then we can introduce some \textit{cut-off operators} $P^{(N)} _\nu (D_\ast)$, $Q^{(N)} _\nu (D_\ast)$ defined as 
	\be 
	P^{(N)} _\nu (z) &\!\!=\!\!& 
	\sum _{k = 0} ^N \frac{1}{k! \, \Gamma \left( \nu + k + 1 \right)} \, \left( \frac{z}{4} \right) ^{k + \frac{\nu}{2}} \, , \\ \label{cut_op_P}
	Q^{(N)} _\nu &\!\!=\!\!& 
	\sum _{k = 0} ^N \frac{1}{k! \, \Gamma \left( \nu + k + 3 \right)} \, \left( \frac{z}{4} \right) ^{k + 1 + \frac{\nu}{2}} \, , \label{cut_op_Q_nu}
	\ee
for which we have that
$$
P_\nu (D_\ast) \, \sigma (t) = \lim _{N \to \infty} P^{(N)} _\nu (D_\ast) \, \sigma (t) \, , \qquad
Q_\nu (D_\ast) \, \varepsilon (t) = \lim _{N \to \infty} Q^{(N)} _\nu (D_\ast) \, \varepsilon (t) \, .
$$

	We can now perform our analysis on the cut-off operators, for which the background formalism is well established, and then extend the results taking the limit for $N \to \infty$.
	
	Given the previous assumptions, we have that $P^{(N)} _\nu (D_\ast) \, \sigma (t)$ and $Q^{(N)} _\nu (D_\ast) \, \varepsilon (t)$ are two sequences of functions that converge almost everywhere on $t \geq 0$. If we further assume that $\exists \alpha , \beta \in L _1 ([0, \infty[)$ such that
	$$ 
	\left| P^{(N)} _\nu (D_\ast) \, \sigma (t) \right| \leq \alpha (t) \, , \qquad
	\left| Q^{(N)} _\nu (D_\ast) \, \varepsilon (t) \right| \leq \beta (t) \, ,
	$$
for all $N \geq 1$ and almost everywhere on $t \geq 0$, then, taking profit of the dominated convergence theorem, we have that	
\be
	\mathcal{L} \left\{ P_\nu (D_\ast) \, \sigma (t) \, ; \, s \right\} 
	&\!\!=\!\!& \lim _{N \to \infty} \mathcal{L} \left\{ P^{(N)} _\nu (D_\ast) \, \sigma (t)  \, ; \, s \right\} \, , \\
	\mathcal{L} \left\{ Q_\nu (D_\ast) \, \varepsilon (t) \, ; \, s \right\} 
	&\!\!=\!\!& \lim _{N \to \infty} \mathcal{L} \left\{ Q^{(N)} _\nu (D_\ast) \, \varepsilon (t)  \, ; \, s \right\} \, . 
	\ee
	This remark then allows us to deal with finite order operators without losing the connection with the infinite order one.

	For sake of brevity, let us focus only on the strain part. The corresponding analysis for the stress part of the equation can be carried out in the same way.
		
First, let us compute the Laplace transform of $ Q^{(N)} _\nu (D_\ast) \, \varepsilon (t)$. Recalling the Laplace transform of the Caputo's fractional derivative
	\be 
	\mathcal{L} \left\{ D_\ast ^\alpha \, f (t) \, ; \, s \right\} = s^\alpha \, \wt{f} (s) - \sum _{k=0} ^{\ceil{\alpha} - 1} f^{(k)} (0+) \, s^{\alpha - k - 1} \, , \quad \alpha \in \R^+ \, ,  
	\ee		
then one can easily infer that
	\begin{small}
	\be
	\mathcal{L} \left\{ Q^{(N)} _\nu (D_\ast) \, \varepsilon (t)  \, ; \, s \right\} = Q ^{(N)} _\nu (s) \, \wt{\varepsilon} (s) - 
	\sum _{n = 0} ^N q_n \sum _{k = 0} ^{\ceil{\nu / 2} + n} \varepsilon ^{(k)} (0+) \, s^{(\nu / 2) + n - k} \, . \label{LT-Q}
	\ee
	\end{small}

Now, one can easily show that
	
	\begin{lem} \label{Lemma-3}
	Let $\alpha \in \R ^+$ then
	\be \label{eq-lemma3}
	\sum _{k=0} ^{\ceil{\alpha} - 1} f^{(k)} (0+) \, s^{\alpha - k - 1} = 
	s^{\{\alpha\} - \delta _\alpha} \sum _{r = 1} ^{\ceil{\alpha}} f^{(\ceil{\alpha} - r)} (0+) \, s^{r - 1} \, ,
	\ee
	where $\{\alpha\}$ is the fractional part of $\alpha$ and $\delta _\alpha$ is such that
	\be
	\delta _\alpha = 
	\left\{
	\begin{aligned}
	& 0, \,\,\, \alpha \in \Z \, ,\\
	& 1, \,\,\, \alpha \notin  \Z \, .
	\end{aligned}
	\right.	
	\ee
	\end{lem}

\begin{rmk}
Notice that for $\alpha \in \R^+$ we have that $\alpha = \{ \alpha \} + \ceil{\alpha} - \delta _\alpha$.
\end{rmk}

\blue{
\begin{proof}	
Starting from the RHS of \eqref{eq-lemma3}, taking profit of the previous remark, one just needs to extract the integer part of $\alpha$ form $s^\alpha$. Then, the LHS of \eqref{eq-lemma3} is obtained by changing the summation index from $k$ to $r=\ceil{\alpha}-k$. Indeed, taking profit of the last remark,
\be
\sum _{k=0} ^{\ceil{\alpha} - 1} f^{(k)} (0+) \, s^{\alpha - k - 1} = s^{\{\alpha\} - \delta _\alpha} \sum _{k=0} ^{\ceil{\alpha} - 1} f^{(k)} (0+) \, s^{\ceil{\alpha} - k - 1} \, ,
\ee
then, defining the index $r = \ceil{\alpha} - k$, we get that
\be 
\sum _{k=0} ^{\ceil{\alpha} - 1} f^{(k)} (0+) \, s^{\ceil{\alpha} - k - 1} = \sum _{r = 1} ^{\ceil{\alpha}} f^{(\ceil{\alpha} - r)} (0+) \, s^{r - 1} \, ,
\ee
that concludes our proof.
\end{proof}	
}

	Then, by means of the result in Lemma \ref{Lemma-3} we can rewrite Eq.~\eqref{LT-Q} as	
	\be
	\begin{split}
	\mathcal{L} \left\{ Q^{(N)} _\nu (D_\ast) \, \varepsilon (t)  \, ; \, s \right\} &= Q ^{(N)} _\nu (s) \, \wt{\varepsilon} (s)  - \frac{s^{\{\nu / 2\} - \delta _{\nu / 2}}}{s} \times \\
	& \times \sum _{n = 0} ^N q_n \sum _{r = 1} ^{\ceil{\nu / 2} + n + 1} s^r \,  \varepsilon ^{(\ceil{\nu / 2} + n + 1 - r)} (0+) 
	 \, . \label{LT-Q1}
	\end{split}
	\ee	
	
	Following a similar argument as the one by Gurtin and Sternberg in \cite{Gurtin-Sternberg}, one can infer that
	\begin{lem} \label{Lemma-4}
	Let $a, M \in \N$ and let $q_n, \, b_n, \, s \in \R, \, \forall n \in \N _0$, then
	\be \label{eq-lemma-4}
	\begin{split}
	\sum _{n = 0} ^N q_n \sum _{r = 1} ^{M + n} s^r \,  b _{M + n - r} = 
	\sum _{k = 1} ^{N + M} s^k \, 
	\Big[ &\Theta ( M - k) \sum _{h = 0} ^N q_h \, b_{M + h - k}  + \\
	& \Theta (k - M - 1) \sum _{h = k - M} ^N q_h \, b_{M + h - k} \Big] \, ,
	\end{split}	
	\ee
	for all $N \in \N$.
	\end{lem}

\blue{
\begin{proof}
	This can be proved by induction. 

First, for $N = 1$ the LHS of Eq.~\eqref{eq-lemma-4} is given by
\be
\begin{split}
\sum _{n = 0} ^1 q_n \sum _{r = 1} ^{M + n} s^r \,  b _{M + n - r} &= 
q_0 \sum _{r = 1} ^{M} s^r \,  b _{M - r} + q_1 \sum _{r = 1} ^{M + 1} s^r \,  b _{M + 1 - r} =\\
&= \sum _{r = 1} ^{M} s^r \left( q_0 \, b _{M - r} + q_1 \, b _{M + 1 - r} \right) + s^{M+1} \, q_1 \, b_0 \, .
\end{split}
\ee
Beside, the RHS of Eq.~\eqref{eq-lemma-4} is given by,
\be
\begin{split}
& \sum _{k = 1} ^{M + 1} s^k \, 
	\Big[ \Theta ( M - k) \sum _{h = 0} ^1 q_h \, b_{M + h - k}  + 
	 \Theta (k - M - 1) \sum _{h = k - M} ^1 q_h \, b_{M + h - k} \Big] =\\
	 &= \sum _{k = 1} ^{M + 1} s^k \, \Big[ \Theta ( M - k) \left(q_0 \, b_{M - k} + q_1 \, b_{M - k} \right) + \Theta (k - M - 1) \, q_1 \, b_0 \Big] =\\
	 &= \sum _{k = 1} ^{M} s^k \left( q_0 \, b _{M - k} + q_1 \, b _{M + 1 - k} \right) + s^{M+1} \, q_1 \, b_0 \, .
\end{split}
\ee
Thus, it is clear that for $N = 1$ both sides of are equal Eq.~\eqref{eq-lemma-4} and, therefore, Eq.~\eqref{eq-lemma-4} is verified for $N=1$.

	Now, let us suppose that Eq.~\eqref{eq-lemma-4} is true for $N = \mathcal{N}$, for $\mathcal{N} > 1$, then for $N = \mathcal{N} + 1$ we have that the LHS is given by
	\be 
	\begin{small}
	\begin{split}
	\!\!\sum _{n = 0} ^{\mathcal{N} + 1}\!\! q_n \sum _{r = 1} ^{M + n} s^r \,  b _{M + n - r} &= 
	\sum _{n = 0} ^{\mathcal{N}} q_n \!\! \sum _{r = 1} ^{M + n} \!\! s^r \,  b _{M + n - r} + q_{\mathcal{N} + 1}\!\! \sum _{r = 1} ^{M + \mathcal{N} + 1} \!\! s^r   b _{M + \mathcal{N} + 1 - r} = \\
	&= \sum _{k = 1} ^{\mathcal{N} + M} s^k \, 
	\Big[ \Theta ( M - k) \sum _{h = 0} ^{\mathcal{N}} q_h \, b_{M + h - k}  +\\
	 &+ \Theta (k - M - 1) \sum _{h = k - M} ^{\mathcal{N}} q_h \, b_{M + h - k} \Big] + \\
	 &+ q_{\mathcal{N} + 1}\!\! \sum _{r = 1} ^{M + \mathcal{N} + 1} \!\! s^r   b _{M + \mathcal{N} + 1 - r} \, .
	\end{split}
	\end{small}
	\ee

	Beside, the RHS of Eq.~\eqref{eq-lemma-4} is given by,
\begin{small}
\be 
\notag
&\!\! \!\!& \sum _{k = 1} ^{\mathcal{N} + 1 + M} s^k \, 
	\Big[ \Theta ( M - k) \sum _{h = 0} ^{\mathcal{N} + 1} q_h \, b_{M + h - k}  +
 	\Theta (k - M - 1) \sum _{h = k - M} ^{\mathcal{N} + 1} q_h \, b_{M + h - k} \Big]  =\\ \notag
&\!\!=\!\!& \sum _{k = 1} ^{\mathcal{N} + M} s^k \, 
	\Bigg[ \Theta ( M - k) \left( \sum _{h = 0} ^{\mathcal{N}} q_h \, b_{M + h - k} + q_{\mathcal{N} + 1} \, b_{M + \mathcal{N} + 1 - k} \right) +\\ \notag
	&\!\!+\!\!& \Theta (k - M - 1) \left( \sum _{h = k - M} ^{\mathcal{N}} q_h \, b_{M + h - k} + q_{\mathcal{N} + 1} \, b_{M + \mathcal{N} + 1 - k} \right) \Bigg] +\\ \notag
	&\!\!+\!\!& s^{M+\mathcal{N} + 1} \, q_{\mathcal{N} + 1} \, b_{0} =\\ \notag
	&\!\!=\!\!& \sum _{k = 1} ^{\mathcal{N} + M} s^k \, 
	\Big[ \Theta ( M - k) \sum _{h = 0} ^\mathcal{N} q_h \, b_{M + h - k}  + 
	 \Theta (k - M - 1) \sum _{h = k - M} ^\mathcal{N} q_h \, b_{M + h - k} \Big] + \\ \notag
	 &\!\!+\!\!& q_{\mathcal{N} + 1}\!\! \sum _{r = 1} ^{M + \mathcal{N} + 1} \!\! s^r   b _{M + \mathcal{N} + 1 - r} \, . \notag
\ee
\end{small}

	Thus, Eq.~\eqref{eq-lemma-4} holds for $N = \mathcal{N} + 1$ and the proof of the induction step is complete. In conclusion, we have that Eq.~\eqref{eq-lemma-4} holds $\forall N \in \N$.
\end{proof}	
}

	We can now use the result in Lemma \ref{Lemma-4} to show that Eq.~\eqref{LT-Q1} can be rewritten as
	\begin{small}
	\be \label{LT-Q2}
	\begin{split}
	\mathcal{L} \Big\{ Q^{(N)} _\nu (D_\ast) \, &\varepsilon (t)  \, ; \, s \Big\} = Q ^{(N)} _\nu (s) \, \wt{\varepsilon} (s)  - \frac{s^{\{\nu / 2\} - \delta _{\nu / 2}}}{s} \times \\
	& \times \sum _{k = 1} ^{\ceil{\nu/2} + N + 1} s^k \, 
	\Bigg[ \Theta ( \ceil{\nu / 2} + 1 - k) \sum _{n = 0} ^N q_n \, \varepsilon ^{(\ceil{\nu / 2} + n + 1 - k)} ( 0 +)  \\
	& + \Theta (k - \ceil{\nu / 2} - 2) \sum _{n = k - \ceil{\nu / 2} - 1} ^N q_n \, \varepsilon ^{(\ceil{\nu / 2} + n + 1 - k)} ( 0 +) \Bigg]
	 \end{split}
	\ee
	\end{small}
	\blue{The proof is quite straightforward, indeed we just need to apply the result in} Lemma \ref{Lemma-4} to the second term in Eq.~\eqref{LT-Q1}, noticing that $M = \ceil{\nu / 2} + 1$ and $b_n = \varepsilon ^{(n)} (0+)$.	
	
	If we perform the same analysis for $P^{(N)} _\nu (D_\ast) \, \sigma (t)$ we get
	\begin{small}
		\be \label{LT-P2}
	\begin{split}
	\mathcal{L} \Big\{ P^{(N)} _\nu (D_\ast) \, &\sigma (t)  \, ; \, s \Big\} = P ^{(N)} _\nu (s) \, \wt{\sigma} (s)  - \frac{s^{\{\nu / 2\} - \delta _{\nu / 2}}}{s} \times \\
	& \times \sum _{k = 1} ^{\ceil{\nu/2} + N} s^k \, 
	\Bigg[ \Theta ( \ceil{\nu / 2} - k) \sum _{n = 0} ^N p_n \, \sigma ^{(\ceil{\nu / 2} + n - k)} ( 0 +)  \\
	& + \Theta (k - \ceil{\nu / 2} - 1) \sum _{n = k - \ceil{\nu / 2}} ^N p_n \, \sigma ^{(\ceil{\nu / 2} + n - k)} ( 0 +) \Bigg] \, .
	 \end{split}
	\ee
		\end{small}	

	Then, by taking the limit for $N \to \infty$ we get $\mathcal{L} \left\{ Q _\nu (D_\ast) \, \varepsilon (t)  \, ; \, s \right\}$ and $\mathcal{L} \left\{ P _\nu (D_\ast) \, \sigma (t)  \, ; \, s \right\}$, respectively.

	Hence, if we now assume that the initial data meet the condition in Eq.~\eqref{initial-condition}, then
	\be	
	\mathcal{L} \left\{ Q _\nu (D_\ast) \, \varepsilon (t)  \, ; \, s \right\} = \mathcal{L} \left\{ P _\nu (D_\ast) \, \sigma (t)  \, ; \, s \right\}
	\ee
	reduces to \eqref{Bessel-stress-strain-laplace}. Then, all the subsequent results discussed in the first part of this section follow immediately from \eqref{initial-condition} and \eqref{Bessel-stress-strain-laplace}.
	
	\blue{
	\subsection{Bessel $\&$ Fractional Maxwell of order $1/2$} \label{sub-1}
	From the previous discussion we have that the constitutive equation for a Bessel model of order $\nu > -1$ can be expressed as in Eq.~\eqref{Bessel-stress-strain-time-regularized}, provided that we assume the validity of the condition in \eqref{initial-condition}.	
	
	Now, following the discussion in \cite{IC-AG-FM-Conf}, it is interesting to see that, even if we are dealing with ordinary infinite order differential equation, as $t \to  0^+$ we have that the stress-strain relation reduces to
	\be 
\left[ 1 + \frac{1}{2 (\nu + 1)} \,  D _{\ast} ^{1/2} \right] \ \sigma (t) 
= 
\frac{1}{2 (\nu + 1)} \, D _{\ast} ^{1/2} \varepsilon (t) \, ,
	\ee 
which is actually Maxwell-like fractional model of order $1/2$ (see \eg \cite{Mainardi_BOOK10, Mainardi-Spada 2011}).	
	
	\subsection{Entire functions and applicability} \label{sub-2}
	Now, let us focus our attention on the regularized operators $\overline{P} _\nu$ and $\overline{Q} _\nu$ defined in terms of the power series \eqref{op_P_nu_reg} and \eqref{op_Q_nu_reg}, respectively. In particular, it is easy to prove the following statement	
	\begin{thm} \label{th-1}
	Let $\overline{P} _\nu (z)$ and $\overline{Q} _\nu (z)$ be the functions defined as in \eqref{op_P_nu_reg} and \eqref{op_Q_nu_reg}. Then we can note that $\overline{P} _\nu (z)$ and $\overline{Q} _\nu (z)$ are entire functions of order $\rho = 1/2$ and type $\sigma = 1$. 
	\end{thm}	
	
\begin{proof}
Let us focus our attention on the function $\overline{P} _\nu (z)$, then the proof for $\overline{Q} _\nu (z)$ follows by the very same procedure. 
 	
 	First, given the definition of $\overline{P} _\nu (z)$, we can immediately infer that it is an entire function because it is defined in terms of a modified Bessel function. 
 	
 	Now, given an entire function represented by the power series $ \sum _{n=0} ^\infty a_n \, z^n$, then the order $\rho$ and the type $\sigma$ of such a function are given by (see \cite{Holland})
\begin{equation*}
\begin{split}
\rho ^{-1} &= \liminf _{n \to \infty} \frac{\log \left( 1 / |a_{n}| \right)}{n \log n}  \, ,\\
(e \, \rho \, \sigma )^{1/\rho } &= \limsup _{n\to \infty} n^{1/\rho } \, |a_{n}|^{1/n} \, .
\end{split}
\end{equation*}

	Thus, for $\overline{P} _\nu (z)$ we have that
	\be
	a_n = \frac{1}{n! \, \Gamma \left( \nu + n + 1 \right) \, 2^{2n}} \, , 
	\ee
	from which it follows that
	\be 
	\log \left( 1 / |a_{n}| \right) = \log n! + \log \Gamma \left( \nu + n + 1 \right) + 2 \, n  \log 2 \, .
	\ee
	Recalling Stirling's approximation, \ie
	\be 
	\log \Gamma (z) = z \log z - z + O (\log z) \, , \qquad |z| \to \infty \, , \,\,\, |\texttt{Arg}(z)| \leq \pi \, ,
	\ee
	then,
	\be 
	\log \left( 1 / |a_{n}| \right) = 2 \, n \log n + O (n) \, , \qquad \mbox{as} \,\, n \to \infty \, .
	\ee
	Hence, for $\overline{P} _\nu (z)$ we can conclude that
	\be 
	\frac{1}{\rho} = \liminf _{n \to \infty} \frac{\log \left( 1 / |a_{n}| \right)}{n \log n} = \liminf _{n \to \infty} 2 + o \left( \frac{1}{\log n} \right) = 2 \, .
	\ee
	Analogously, it is easy to see that
	$$ \limsup _{n\to \infty} n^{2} \, |a_{n}|^{1/n} = \frac{e^2}{4} \, , $$
	from which we can infer that
	$$ \sigma = 1 \, , $$
	result that concludes our proof.
\end{proof}	
	
\begin{rmk}
We could have come to the same results in Theorem \ref{th-1}, alternatively, by recalling that the functions \eqref{op_P_nu_reg} and \eqref{op_Q_nu_reg} are basically a rescaled version of the modified Bessel functions. Now, the latter can be seen as particular realizations of the Mittag-Leffler functions, for which a complete discussion of the order and type has been carried out by Kiryakova in \cite{Kiryakova}.
\end{rmk}	
	
	This result is particularly interesting for two distinct reasons.
	 	
	First, the case in which an infinite order differential operator is defined in terms of a generating power series corresponding to an entire function has been studied by several authors, see \eg \cite{Cha-etal, Sikkema}. In particular, according to these papers, in this case it seems to be easier to prove the applicability of such operators to certain subsets of the space of entire functions. For example, in \cite{Cha-etal} Cha et al  have been able to provide a sufficient condition for two entire functions, $f(z)$ and $g(z)$, to be such that the combination $\sum _n f^{(n)} (0) \, D^n g(z) / n!$ represents an entire function. Moreover, they have also deduced the growth of the resulting function, induced by the types of both $f(z)$ and $g(z)$. 

	Secondly, it is interesting to note that the (fractional) order of asymptotic constitutive equation \eqref{Bessel-stress-strain-time-regularized}, as $t \to 0^+$, matches exactly with the order of the entire functions defining the differential operators. Whether there is a formal connection between these two results is still unknown, but it is surely worth of further studies.
	
}

		
\section{Conclusions}
	After a brief review of the basic formalism of linear viscoelasticity, we have discussed some properties of viscoelastic models defined in terms of constitutive equations involving infinitely many derivatives. In particular, we have considered, as a working example, the recently developed Bessel models of linear viscoelasticiy. 
	
	In Section \ref{Sec-const-eq}, we have been able to connect the constitutive equation for a Bessel model with the theory of \textit{ordinary} infinite order differential equations. We have also determined the form of the initial data that guarantees the mutual consistency of both the integral and the differential form of the stress-strain relation. This condition also allows for a constitutive equation, in the Laplace domain, that does not involve the explicit contribution of initial data. 
	Finally, in Section \ref{sub-1} and \ref{sub-2} we have observed that, even thought the stress-strain relation for a Bessel body involves infinitely many \textit{ordinary} derivatives, the asymptotic behaviour as $t \to 0^+$ is effectively described by a fractional Maxwell model of order $1/2$. We have also pointed out that the order of the asymptotic (fractional) model is the same as the order of the entire functions that define the regularized operators \eqref{op_P_nu_reg} and \eqref{op_Q_nu_reg}.

\section*{Acknowledgements}
	The work of the author has been carried out in the framework of the activities of the National Group of Mathematical Physics (GNFM, INdAM).
 	
 	I am particularly grateful to Prof. Francesco Mainardi for all of the support and mentorship he has shown me over the past eight years and to Prof. Roberto Casadio for helpful comments and discussions, and also for his patience and constant support. I would also like to extend my thanks to I. Colombaro, V. A. Diaz, Prof. M. Fabrizio, R. Garra, Prof. A. Kamenchtchik, \blue{Prof. V. Kiryakova, Prof. S. Rogosin} and Prof. T. Ruggeri for many helpful discussions.


\vskip 1 cm

 
 \end{document}